\documentclass[a4paper,english,pagebackref]{article}

\usepackage[utf8]{inputenc}
\usepackage{amsmath,amssymb,xspace,amsthm}
\usepackage{paralist}
\usepackage{authblk}
\usepackage{booktabs}
\usepackage{tabularx}
\usepackage{multirow}
\usepackage{a4wide}
\usepackage[numbers,sort&compress]{natbib}

\makeatletter
\def\NAT@spacechar{~}%
\makeatother
\usepackage{tikz}
\usetikzlibrary{decorations.pathmorphing}
\usetikzlibrary{decorations.pathreplacing}
\usetikzlibrary{shapes,automata}
\usetikzlibrary{calc}

\usepackage{color}
\usepackage{graphicx}
\definecolor{mygrey}{rgb}{0.9,0.9,0.9}
\definecolor{darkgreen}{RGB}{0,100,0}

\usepackage[pdfdisplaydoctitle,menucolor=orange!40!black,filecolor=magenta!40!black,urlcolor=blue!40!black,linkcolor=red!40!black,citecolor=green!40!black,colorlinks]{hyperref}

\usepackage[sort&compress,nameinlink,noabbrev,capitalize]{cleveref}

\usepackage[T1]{fontenc}
\usepackage{thmtools}
\usepackage{dsfont}
\usepackage{etoolbox}

\theoremstyle{plain}
\newtheorem{theorem}{Theorem}
\newtheorem{corollary}{Corollary}
\newtheorem{lemma}{Lemma}
\newtheorem{proposition}{Proposition}

\theoremstyle{remark}

\theoremstyle{definition}
\newtheorem{definition}{Definition}
\newtheorem{constr}{Construction}

\crefname{constr}{Construction}{Constructions}
\crefname{step}{Step}{Steps}
\crefname{observation}{Observation}{Observations}
\crefname{proposition}{Proposition}{Propositions}
\crefname{theorem}{Theorem}{Theorems}
\Crefname{theorem}{Thm.}{Thms.}
\crefname{remark}{Remark}{Remarks}
\crefname{prop}{Property}{Properties}
\creflabelformat{prop}{(#2#1#3)}

\newcommand{\decprob}[3]{
	\begin{center}
		\begin{minipage}{0.995\textwidth}
			\noindent
			\textsc{#1}\\
			\setlength{\tabcolsep}{3pt}
			\begin{tabularx}{\textwidth}{@{}lX@{}}
					\normalsize \textbf{Input:} 		& \normalsize #2 \\
					\normalsize \textbf{Question:} 	& \normalsize #3
				\end{tabularx}
		\end{minipage}
	\end{center}
}

\DeclareMathOperator{\dist}{dist}

\newcommand{\no}{\textsc{no}\xspace}
\newcommand{\yes}{\textsc{yes}\xspace}

\newcommand{\N}{\mathds{N}}

\newcommand{\Z}{\mathds{Z}}

\DeclareMathOperator{\col}{col}
\newcommand{\floor}[1]{\mathop{\lfloor #1 \rfloor}}

\newcommand{\si}{Connected Subgraph Isomorphism\xspace}
\newcommand{\siAcr}{\textsc{CSI}\xspace}
\newcommand{\siTsc}{\textsc{\si}\xspace}

\newcommand{\hsi}{$H$-Subgraph Isomorphism\xspace}
\newcommand{\hsiAcr}{\textsc{$H$-SI}\xspace}
\newcommand{\hsiTsc}{\textsc{\hsi}\xspace}

\newcommand{\calK}{\mathcal{K}}
\newcommand{\calD}{\mathcal{D}}
\newcommand{\qedconstr}{\hfill$\diamond$}
\newcommand{\eps}{\ensuremath{\varepsilon}}

\newcommand{\thetitle}{Kernelization Lower Bounds for \\ Finding Constant-Size Subgraphs}%

\title{\thetitle{}}
\date{\vspace{-20pt}}

\author[1]{Till~Fluschnik\thanks{Supported by the DFG, projects DAMM (NI~369/13-2) and TORE (NI~369/18).}$ ^{,}$}
\author[2]{George~B.~Mertzios\thanks{Partially supported by the EPSRC grant EP/P020372/1.}$ ^{,}$}
\author[1]{Andr\'{e}~Nichterlein\thanks{Supported by a postdoc fellowship of DAAD while at Durham University.}$ ^{,}$}

\affil[1]{\small{Algorithmics and Computational Complexity, Faculty IV, TU~Berlin, Germany, \texttt{\{till.fluschnik,andre.nichterlein\}@tu-berlin.de}}}
\affil[2]{\small{Department of Computer Science, UK, \texttt{george.mertzios@durham.ac.uk}}}

\hypersetup{pdftitle={\thetitle{}
}, pdfauthor={TF,GM,AN
}}

\begin{document}

\maketitle
\begin{abstract}
  Kernelization is an important tool in parameterized algorithmics.
  Given an input instance accompanied by a parameter, the goal is to compute in polynomial time an equivalent instance of the same problem such that the size of the reduced instance only depends on the parameter and not on the size of the original instance.
  In this paper, we provide a first conceptual study on limits of kernelization for several \emph{polynomial-time} solvable problems.
  For instance, we consider the problem of finding a triangle with negative sum of edge weights parameterized by the maximum degree of the input graph.
  We prove that a linear-time computable strict kernel of truly subcubic size for this problem violates the popular APSP-conjecture.
\end{abstract}

\section{Introduction}\label{sec:intro}

Kernelization is the main mathematical concept for provably efficient preprocessing of computationally hard problems. 
This concept has been extensively studied (see, e.g.,~\cite{FominS14,GN07,Kra14,LokshtanovMS12}) and it has great potential for delivering practically relevant algorithms~\cite{Iwata17,Wei98}.
In a nutshell, the aim is to significantly and efficiently reduce a given instance of a parameterized problem to its ``computationally hard core''. 
Formally, given an instance~$(x,k)\in \{0,1\}^{\ast}\times \N$ of a parameterized problem $L$, a \emph{kernelization} for $L$ is an algorithm that computes in polynomial time an instance~$(x',k')$, called kernel, such that (i)~$(x,k)\in L \iff (x',k')\in L$ and (ii)~$|x'|+k' \leq f(k)$, for some computable function $f$. 
Although studied mostly for NP-hard problems, it is natural to apply this concept also to polynomial-time solvable problems as done e.g.~for finding maximum matchings~\cite{MNN17}.
It is thus also important to know the \emph{limits} of this concept. 
In this paper we initiate a systematic approach to derive \emph{kernelization lower bounds} for problems in P.
We demonstrate our techniques at the example of subgraph isomorphism problems where the sought induced subgraph has constant size and is connected.

When kernelization is studied on NP-hard problems (where polynomial running times are considered computationally ``tractable''), the main point of interest becomes the \emph{size}~$f(k)$ of the kernel with respect to the parameter $k$. 
In particular, from a theoretical point of view, one typically wishes to minimize the kernel size to an---ideally---polynomial function $f$ of small degree.  
As every decision problem in P admits a kernelization which simply solves the input instance and produces a kernel of size~$O(1)$ (encoding the \yes/\no answer), it is crucial to investigate the \emph{trade-off} between (i)~the size of the kernel and (ii)~the running time of the kernelization algorithm. 
The following notion captures this trade-off:
An \emph{$(a,b)$-kernelization} for a parameterized problem~$L$ is an algorithm that, given any instance~$(x,k)\in\{0,1\}^{\ast}\times \N$, computes in $O(a(|x|))$~time an instance~$(x',k')$ such that
	\begin{inparaenum}[(i)]
		\item $(x,k)\in L$ $\iff$ $(x',k')\in L$ and
		\item $|x'|+k'\in O(b(k))$.
	\end{inparaenum}

Kernelization for problems in~P is part of the recently introduced framework ``FPT in P''~\cite{GMN17-TCS}.
This framework is recently applied to investigate parameterized algorithms and complexity for problems in P~\cite{AWW16,FKMNNT17,FominLPSW_Matching-Treewidth_SODA17,GMN17-TCS,MNN17}.  
Studying lower bounds for kernelization for problems in~P is---as it turns out---strongly connected to the active research field of \emph{lower bounds} on the running times of polynomial-time solvable problems~(see, e.g,~\cite{AGW15,AVW14,AWW16,Bringmann14}). 
These running time lower bounds rely on popular conjectures like the Strong Exponential Time Hypothesis (SETH)~\cite{IP01,IPZ01} or the \textsc{3SUM}-conjecture~\cite{GajentaanO95}, for instance.

In contrast to NP-hard problems, only little is known about kernelization lower bounds for problems in P.
To the best of our knowledge all known kernelization lower bounds follow trivially from the corresponding lower bounds of the running time: 
For instance, assuming SETH, it is known that (i) the hyperbolicity and (ii) the diameter of a graph cannot be computed in $2^{o(k)}\cdot n^{2-\eps}$ time for any $\eps>0$, where $k$ is (i) the vertex cover number and (ii) the treewidth of the graph~\cite{FKMNNT17,AWW16}. 
This implies that both problems do not admit an~$(n^{2-\eps},2^{o(k)})$-kernelization---a kernel with $2^{o(k)}$ vertices computable in~$O(n^{2-\eps})$ time---since such a kernelization yields an algorithm running in~$O(2^{o(k)} + n^{2-\eps})$~time.

In this paper we initiate a systematic approach to derive kernelization lower bounds for problems in P for a---very natural---special type  of kernels.

\begin{definition}[strict $(a,b)$-kernelization]
 \label{def:strict-kernel}
	A \emph{strict $(a,b)$-kernelization} for a parameterized problem~$L$ is an algorithm that given any instance~$(x,k)\in\{0,1\}^{\ast}\times \N$ computes in~$O(a(|x|))$ time an instance $(x',k')$ such that
	\begin{inparaenum}[(i)]
		\item $(x,k)\in L$ $\iff$ $(x',k')\in L$,
		\item $|x'|+k'\in O(b(k))$, and 
		\item $k'\leq k$.
	\end{inparaenum}
\end{definition}

Chen et al.~\cite{CFM11} %
introduced a framework to exclude strict kernels for NP-hard problems, assuming that P${}\neq{}$NP. 
Fernau et al.~\cite{FFHKMN16} applied the framework to a wide variety of FPT problems and studied it on ``less'' strict kernelizations. 
The framework~\cite{CFM11,FFHKMN16} is based on the notion of \emph{(strong) diminishers}:

\begin{definition}[$a$-diminisher]
\label{def:diminisher}
	An \emph{$a$-diminisher} for a parameterized problem~$L$ is an algorithm that given any instance~$(x,k)\in\{0,1\}^{\ast}\times \N$ in $O(a(|x|))$~time either decides whether $(x,k)\in L$ or computes an instance~$(x',k')$ such that
	\begin{inparaenum}[(i)]
		\item $(x,k)\in L$ $\iff$ $(x',k')\in L$, and
		\item $k'< k$.
	\end{inparaenum}
	A \emph{strong $a$-diminisher} for $L$ is an $a$-diminisher for~$L$ with $k'< k / c$ for some constant~$c > 1$.
\end{definition}
\paragraph{Our Contributions.}

We adapt the diminisher framework~\cite{CFM11,FFHKMN16} to prove kernelization lower bounds for problems in P. 
Our results concern the \textsc{$H$-Subgraph Isomorphism} \emph{($H$-SI)} problem\footnote{The \textsc{$H$-Subgraph Isomorphism} asks, given an undirected graph~$G=(V,E)$, whether $G$ contains~$H$ as a subgraph.} for \emph{constant-sized connected} graphs~$H$. 
As a running example, we focus on the fundamental case where~$H$ is a triangle and we present diminishers (along with kernelization lower bounds) for the following weighted and colored variants of the problem:

\decprob{Negative Weight Triangle (NWT)}
{An undirected graph~$G$ with edge weights~$w\colon E(G)\to\Z$.}
{Is there a triangle~$T$ in~$G$ with $\sum_{e\in E(T)} w(e)<0$?}

\decprob{Triangle Collection (TC)}
{An undirected graph~$G$ with surjective coloring~$\col:V(G)\to[f]$.}
{Does there for all color-triples $C\in\binom{[f]}{3}$ exist a triangle with vertex set~$T=\{x, y, z\}$ in $G$ 
such that~$\col(T)=C$?}

\newcommand{\smtab}[1]{\scriptsize#1}
\renewcommand{\arraystretch}{1}
\begin{table*}[t]
  \setlength{\tabcolsep}{10pt}
  \centering
  \caption
  {
    Overview of our results. Here, $k$ is interchangeably the order of the largest connected component, the degeneracy, or the maximum degree.
  }%
  \begin{tabularx}{0.98\textwidth}{@{}cXX}  \toprule
					& \textsc{Negative Weight Triangle} (NWT) 	& \textsc{Triangle Collection (TC)}  \\\midrule
		lower 
					& \multicolumn{2}{c}{No strict $(n^{\alpha},k^{\beta})$-kernelization with~$\alpha,\beta\geq1$ and $\alpha \cdot \beta < 3$, assuming:}  \vspace{0.0cm}\\
		bounds			& \multicolumn{1}{c}{\multirow{2}{*}{the APSP-conjecture.}}				& \multicolumn{1}{|c}{the SETH, APSP-, or}\\
		(\Cref{thm:mainklb})	& 						&  \multicolumn{1}{|c}{3SUM-conjecture.}\\\midrule
		kernel			& \multicolumn{2}{c}{Strict $(n^{(3+\eps)/(1+\eps)},k^{1+\eps})$-kernelization for every~$\eps>0$,}  \vspace{0.0cm} \\
		(\Cref{thm:pl-fpt-in-p-implies-kernel})					& \multicolumn{2}{c}{e.g.\ strict $(n^{5/3},k^{3})$-kernelization.}  \vspace{0.0cm}\\
   \bottomrule	
  \end{tabularx}
  \label{tab:results}
\end{table*}

NWT and TC are conditionally hard:
If NWT admits a \emph{truly} subcubic algorithm---that is, with running time $O(n^{3-\eps})$, $\eps>0$---then APSP also admits a truly subcubic algorithm, breaking the APSP-conjecture~\cite{WW10}.
A truly subcubic algorithm for TC breaks the SETH, the 3SUM-, and the APSP-conjecture~\cite{AWY15}. 

For both NWT and TC we consider three parameters (in decreasing order): 
(i)~order (that is, the number of vertices) of the largest connected component, (ii)~maximum degree, and (iii)~degeneracy. 
We prove that both NWT and TC admit a strong linear-time diminisher for all these three parameters.
Together with the conditional hardness, we then obtain lower bounds on strict kernelization.
Our results are summarized in~\cref{tab:results}.

Complementing our lower bounds, we prove a strict $(n^{5/3},k^{3})$-kernelization for NWT and TC ($k$ being any of the three aforementioned parameters) and a strict~$(n\cdot\Delta^{\floor{c/2}+1},\Delta^{\floor{c/2}+1})$-Turing kernelization for \textsc{$H$-Subgraph Isomorphism} when parameterized by the maximum degree~$\Delta$, where~$c=|V(H)|$.

\paragraph{Notation and Preliminaries.}

We use standard notation from parameterized complexity
\cite{CFK+15}
and graph theory
\cite{Die10}.
For an integer~$j$, we define~$[j]:=\{1,\ldots,j\}$.

\section{Frameworks to Exclude Polynomial Kernelizations}%
\label{sec:frameworks}
We briefly recall the existing frameworks to exclude (strict) polynomial-size kernels for NP-hard problems.
We further discuss the difficulties that appear when transferring these approaches to polynomial-time solvable problems.

\subsection{Composition Framework}

The frequently used (cross-)composition frameworks~\cite{BDFH09jcss,FS11,BJK14} are the tools to exclude polynomial-size problem kernels under the assumption NP${}\subseteq{}$coNP/poly.
We describe the intuitive idea behind these frameworks on the example of \siTsc (\siAcr): Given two undirected graphs~$G=(V,E)$ and~$H=(W,F)$ where~$H$ is connected, decide whether~$G$ contains~$H$ as a subgraph?
We consider \siAcr{} parameterized by the order~$k$ of the largest connected component of the input graph.

Assume \siAcr has a kernel of size~$O(k^c)$ for some constant~$c$.
Let~$(G_1, H), (G_2, H), \ldots, (G_\ell, H)$ be several instances of \siAcr with the same connected graph~$H$.
Clearly, the graph~$G$ obtained by taking the disjoint union of all~$G_i$ contains~$H$ if and only if some~$G_i$ contains~$H$.
Furthermore, the parameter of~$G$ is~$\max_{i \in [\ell]}\{|V(G_i)|\}$.
By choosing~$\ell = k^{c+1}$, it follows that kernelizing the instance~$(G,H)$ yields an instance of size less than~$\ell$, that is, less bits than instances encoded in~$G$.
Intuitively, this means that the kernelization algorithm had to solve at least one of the instances~$(G_i,H)$ in polynomial time.
Since \siAcr is NP-complete, this is believed to be unlikely.

The composition framework formalizes this intuitive approach.
If one uses the original proof strategy based on a result of Fortnow and Santhanam~\cite[Theorem 3.1]{FS11}, then one arrives at the following intermediate statement:
``If \siAcr parameterized by the order~$k$ of the largest connected component admits an $O(n^c)$-time computable $O(k^{c'})$-size kernel, then $\overline{\text{\siAcr}}{}\in{}$NTIME$(n^{c'(c+1)})/n^{c+1}$.'' 
(Here, $\overline{\text{\siAcr}}$ denotes the complement of \siAcr.)
This means that e.\,g.\ a linear-time linear-size kernel would imply a nondeterministic quadratic-time algorithm using~$n^2$ advice to detect no-instances of \siAcr.

The next step in the proof strategy is to exploit the NP-completeness of \siAcr.
Thus, one can reduce any problem in coNP in polynomial time to $\overline{\text{\siAcr}}$. 
Furthermore, since~$c,c' \in O(1)$ one can deduce from the statement $\overline{\text{\siAcr}}{}\in{}$NTIME$(n^{c'(c+1)})/n^{c+1}$, that~NP${}\subseteq{}$coNP/poly, which in turn implies that the polynomial hierarchy collapses~\cite{Yap83}.
Thus, such a kernel is unlikely. 

There are some issues when adapting these frameworks for problems in P.
We discuss the issues using the \hsiTsc~(\hsiAcr) problem for constant-sized connected~$H$.
Adapting the proofs of~Bodlaender et al.~\cite{BDFH09jcss} and Fortnow and Santhanam~\cite{FS11} for \hsiAcr leads to the following:
``If \hsiAcr parameterized by the order~$k$ of the largest connected component admits an $(n^c,k^{c'})$-kernelization, then $\overline{H\text{-SI}}{}\in{}$NTIME$(n^{c'(c+1)})/n^{c+1}$.''
Since there exists a trivial $O(n^{|H|+1})$-time brute-force algorithm for \textsc{\hsiAcr}, there also exist trivial polynomial-time computable kernels for \textsc{\hsiAcr}.
Hence, we have to stick with specifically chosen~$c$ and~$c'$ (with $c \cdot c' < |H|$).
Furthermore, we cannot transfer these results easily to other problems in P due to the lack of a suitable completeness theory (\textsc{\hsiAcr} belongs to~P).

One drawback of the composition approach for any problem~$L$ in P is the lack of clarity on the assumption's ($\overline{L} \notin{}$NTIME$(n^{c'(c+1)})/n^{c+1}$) reasonability.
Moreover, due to a missing equivalent to the NP-completeness theory, the assumption bases on specific problems and not on complexity classes.

\subsection{Strict Kernelization and Diminishers}

Chen et al.~\cite{CFM11} introduced a framework to exclude \emph{strict} kernelization, that is, kernelization that do not allow an increase in the value of the parameter in the obtained kernel instance.
This framework builds on the assumption P${}\neq{}$NP and can be easily adapted to exclude strict kernels for polynomial-time solvable problems.
Recall that for problems in P, 
both the size of the kernel \emph{and} the kernelization running time are important.

We use the problem NWT, parameterized by the order~$k$ of the largest component, as a running example.
Recall that the unparameterized version of this problem is as hard as APSP~\cite{WW10}.
Now the question is whether there is a strict $(n+m, k)$-kernelization for NWT parameterized by the size~$k$ of the largest component.
Given an input~$(G=(V,E),k)$ of NWT such a strict kernelization produces in~$O(n+m)$~time an equivalent instance~$(G',k')$ with~$|G'|+k'\in O(k)$ and~$k'\leq k$.
We will prove that such a strict kernelization would yield a truly subcubic algorithm for APSP.
Our argument relies on the key concept of an \emph{$a$-diminisher} (see \cref{def:diminisher}).

In~\cref{sec:lower-bounds}, we provide a strong $(n+m)$-diminisher for NWT($k$).
Now assume that there is a strict $(n+m, k)$-kernelization for NWT($k$).
The basic idea of the whole approach is to \emph{alternately} apply the diminisher and the kernel. %
Intuitively, one application of the diminisher will halve the size of the connected components at the cost of increasing the size of the instance. 
In turn, the strict kernel bounds the size of the instance in~$O(k)$ without increasing~$k$.
Thus, after $\log(k)$ rounds of applying a strong diminisher and a strict kernel we arrive at an instance~$I$ with constant size connected components.
Then, we can use even a simple brute-force algorithm to solve each connected component in~$O(1)$ time which gives an~$O(n+m)$ time algorithm to solve the instance~$I$.
Altogether, with~$\log k \le \log n$ rounds, each requiring~$O(n+m)$ time, we arrive at an $O((n+m)\log n)$-time algorithm for NWT.
This implies a truly subcubic algorithm for APSP, thus contradicting the APSP-conjecture~\cite{WW10}. 
Formalizing this idea of interleaving diminisher and strict kernel yields the following.

\begin{theorem}
	\label{thm:kernel+dimin=alg}
	Let $L$ be a parameterized problem with parameter $k$ such that each instance with parameter $k\leq c$ for some constant $c>0$ is a trivial instance of~$L$.
	If $L$ with parameter $k$ admits a strict $(a,b)$-kernelization and an $a'$-diminisher (a strong $a'$-diminisher), then any instance $(x,k)$ is solvable in $O(k\cdot (a(a'(b(k)))+a(|x|))$~time (in $O(\log k \cdot (a(a'(b(k)))+a(|x|))$~time).
\end{theorem}
\begin{proof}
 Let $(x,k)$ be an instance of $L$ with parameter $k$.
 Let $\calK$ be a strict $(a,b)$-kernelization and $\calD$ be a $a'$-diminisher.
 Apply $\calK$ on $(x,k)$ to obtain an instance $(x',k')$ with $|x'|+k'\leq b(k)$ and $k'\leq k$.
 This step requires $O(a(|x|))$ time.
 Next, until $k'\leq c$, apply $\calK\circ\calD$ iteratively.
 Each iteration requires at most $O(a'(b(k)))$ time for the $a'$-diminisher, and since the size of the resulting instance is upper-bounded by $O(a'(b(k)))$, 
the subsequent kernelization requires $O(a(a'(b(k))))$ time.
 Since in each iteration, the value of $k'$ decreases by one, there are at most $k$ iterations.
 (If~$\calD$ is a strong $a'$-diminisher, then the number of rounds is~$\log_c k = O(\log k)$.)
 Finally, if $k'\leq c$, the algorithm decides the obtained instance in constant time.
 Hence, the algorithm requires $O(k\cdot a(a'(b(k)))+a(|x|))$ ($O(\log k\cdot a(a'(b(k)))+a(|x|))$) time to decide $(x,k)$.
\end{proof}

It is important to note here a subtle difference between strict kernels and the ``classical'' kernels 
(i.e.~where the obtained parameter is allowed to be upper-bounded by any function in the parameter of the input instance). 
In the context of classical kernels we can draw kernelization upper and lower bound conclusions 
by classifying the various parameters in a (partial) hierarchy, 
according to which parameter is (asymptotically) smaller or larger than the other. 
That is, if there exists a polynomial-sized kernel for a ``small'' parameter $k$, 
then there trivially also exists a polynomial-sized kernel for a ``large'' parameter $k'$. 
Similarly, if a problem does not admit a polynomial-sized kernel for $k'$ 
(assuming that NP${}\nsubseteq{}$coNP/poly), 
then this problem also does not admit a polynomial-sized kernel for $k$. 
However, such a hierarchy of the parameters does not imply---in principle---anything about the 
existence or non-existence of certain strict kernels. 

Indeed, consider two parameters $k$ and $k'$ for a problem $L$, where $k' > k$. 
Assume that $L$ admits a diminisher for parameter parameter $k'$; 
that is, $L(k')$ excludes a certain strict kernel (assuming some complexity-theoretic conjecture such as, 
for example, APSP). 
Then, the value of the parameter $k'$ in the instance produced by this diminisher is 
strictly smaller than the value of~$k'$ in the input instance (see \cref{def:diminisher}). 
However, as the \emph{size} of the instance produced by the diminisher typically increases, 
it might be the case that the value of $k$ in this new instance is \emph{larger} than the value of $k$ 
in the input instance. In such a case, the existence of a diminisher for the large parameter $k'$ 
does not immediately imply a diminisher for the small parameter $k$, and thus a strict kernel for $k$ might---in principle---exist, although no strict kernel exists for $k'$.
%
%

%
\subsection{Reductions for Transferring Kernels}

For NP-complete problems, it is easy to transfer polynomial kernelization results using the following type of reductions~\cite{BTY11}:
Given two parameterized problems~$L,L' \subseteq \Sigma^* \times \N$, a \emph{polynomial parameter transformation} from~$L$ to~$L'$ is a polynomial-time computable mapping~$f\colon\Sigma^* \times \N \rightarrow \Sigma^* \times \N$ that maps every instance~$(x,k)$ to an instance~$(x',k')$ such that 
\begin{inparaenum}[(i)]
	\item $(x,k) \in L \iff (x',k') \in L'$ and 
	\item $k' \le k^{O(1)}$.
\end{inparaenum}

To see that this is the ``correct'' notion of reduction consider the case that~$L'$ admits a polynomial kernel and its unparameterized version~$\bar{L'}$ admits a polynomial-time reduction to~the unparameterized version~$\bar{L}$ of~$L$.
If there is a polynomial parameter transformation from~$L$ to $L'$, then~$L$ also has a polynomial kernel:
Let~$(x,k)$ be the instance of~$L$. 
Then, using the polynomial parameter transformation from~$L$ to~$L'$ we compute in polynomial time the instance~$(x_1,k_1)$ for~$L'$ with~$k_1\leq k^{O(1)}$.
Next, we use the polynomial kernelization for~$L'$ to obtain the kernel~$(x_2,k_2)$ for~$L'$ such that~$|x_2| \le k_2^{O(1)}$ and~$k_2 \le k_1^{O(1)}\leq k^{O(1)}$.
The transformation yields~$|x''| \le k^{O(1)}$.
By assumption, there is a polynomial-time reduction from~$\bar{L'}$ to~$\bar{L}$ which we can use to transfer the kernel to~$L$, as the obtained instance for~$L$ is still of size~$k^{O(1)}$ and thus is desired polynomial kernel.
Consequently, if~$L$ does not admit a polynomial kernel (e.g.~under some complexity theoretical assumption), then~$L'$ does not as well.

There are two issues when using the strategy of \emph{polynomial parameter transformations} to transfer results of \cref{thm:kernel+dimin=alg} along polynomial-time solvable problems:
First, we need to require the transformation to be computable ``fast'' enough and that the parameter does not increase~($k' \le k$).
Second, in order to transfer a strict kernel we need to show a reverse transformation from~$L'$ to~$L$ which again is computable ``quick'' enough and does not increase the parameter.
Hence, we essentially need to show that the two problems~$L$ and~$L'$ are \emph{equivalent} under these restrictive transformations.

\section{Kernelization Lower Bounds via Diminishers} 
  \label{sec:lower-bounds}

In this section, we present diminishers for \hsiTsc (\hsiAcr) for connected~$H$ with respect to the structural parameters 
\begin{inparaenum}[(i)]
	\item order~$\ell$ of the largest connected component,
	\item maximum degree~$\Delta$, and
	\item degeneracy~$d$.
\end{inparaenum}
Observe that~$d \le \Delta \le \ell$ in every graph. 
These lead to our following main result.

\begin{theorem}
  \label{thm:mainklb}
  If NWT (TC) parameterized by~$k$ being the 
  \begin{inparaenum}[(i)]
	\item order~$\ell$ of the largest connected component,
	\item maximum degree~$\Delta$, or
	\item degeneracy~$d$
  \end{inparaenum}
  admits a strict $(n^\alpha, k^\beta)$-kernel for constants~$\alpha, \beta \ge 1$ with~$\alpha \cdot \beta < 3$, then the APSP-conjecture (the SETH, the 3SUM-, and the APSP-conjecture) breaks.
\end{theorem}
\subsection{Parameter Order of the Largest Connected Component}
In the following, we prove a linear-time strong diminisher regarding the parameter order of the largest connected component for problems of finding constant-size subgraphs (with some specific property).
The idea behind our diminisher is depicted as follows: for each connected component, partition the connected component into small parts and then take the union of not too many parts to construct new (connected) components (see~\cref{fig:NWT-diminisher} for an illustration of the idea with~$H$ being a triangle).

\begin{constr}
  \label{constr:constsizedim}
  Let~$H$ be an arbitrary but fixed connected constant-size graph of order~$c>1$.
  Let~$G=(V,E)$ be a graph with the largest connected component being of order~$\ell$.
  First, compute in $O(n+m)$~time the connected components~$G_1, \ldots, G_r$ of~$G$.
  Then, construct a graph~$G'$ as follows.
  
  Let~$G'$ be initially the empty graph.
  If~$\ell\leq 4c$, then set~$G'=G$.
  Otherwise, if~$\ell>4c$, then construct~$G'$ as follows.
  For each connected component~$G_i = (V_i, E_i)$, do the following.
  If the connected component~$G_i = (V_i, E_i)$ is of order at most~$\ell/2$, then add~$G_i$ to~$G'$.
  Otherwise, if~$n_i:=|V_i| > \ell/2$, then we partition~$V_i$ as follows.
  Without loss of generality let $V_i$ be enumerated as $V_i=\{v_i^1,\ldots,v_i^{n_i}\}$.
  For every $p\in\{1,\ldots,4c\}$, define $V_i^p:=\{v_i^q\in V_i\mid q\bmod 4c=p-1\}$.
  This defines the partition~$V_i=V_i^1\uplus\cdots\uplus V_i^{4c}$.
  Then, for each $\{a_1,\ldots,a_c\}\in \binom{[4c]}{c}$, add the graph~$G[V_i^{a_1} \cup \ldots \cup V_i^{a_c}]$ to~$G'$.
  This completes the construction.
  \qedconstr{}
\end{constr}

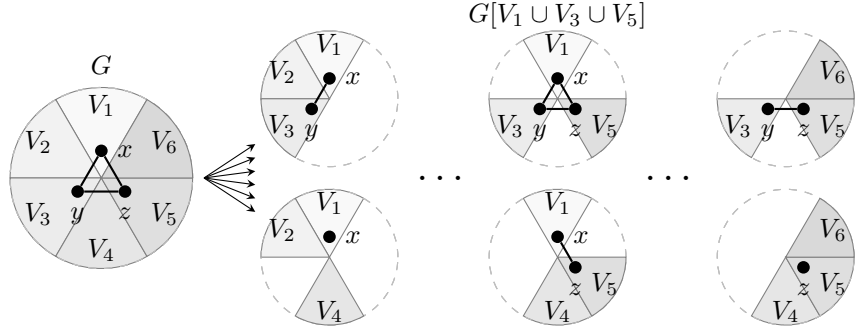
\begin{figure}[t]
 \centering
 
\begin{tikzpicture}[scale=0.6]
  
  \def\pieces{6}
  \def\colorval{5}
  \def\percent{100/\pieces}

  \newcommand{\cake}[7]{
    
    \def\xsh{#2}
    \def\ysh{#3}
    
    \foreach \x in {#1}{
	    \def\angle{\x*\percent*3.6}
	    \pgfmathparse{\colorval*\x} 
	    \xdef\val{\pgfmathresult}  
	    \draw[dashed,color=lightgray] (\xsh,\ysh) circle  (\radius);
	    \draw[fill={gray!\val},draw={gray}] (\xsh,\ysh) -- ($(\angle:\radius)+(\xsh,\ysh)$)
	      arc  (\angle:\angle+\percent*3.6:\radius) -- cycle;
	    \node at ($(\angle+0.5*\percent*3.6:0.8*\radius)+(\xsh,\ysh)$) {$V_{\x}$};
    }
    \foreach \x in {#1}{
	\def\angle{\x*\percent*3.6}
	\ifnum\x=#4\node (x) at ($(\angle+0.5*\percent*3.6:0.3*\radius)+(\xsh,\ysh)$)[fill, circle,scale=1/2,label=0:$x$]{};\fi{}
	\ifnum\x=#5\node (y) at ($(\angle+0.5*\percent*3.6:0.3*\radius)+(\xsh,\ysh)$)[fill, circle,scale=1/2,label=-90:$y$]{};\fi{}
	\ifnum\x=#6\node (z) at ($(\angle+0.5*\percent*3.6:0.3*\radius)+(\xsh,\ysh)$)[fill, circle,scale=1/2,label=-90:$z$]{};
	  
	\fi{}
    }
    \ifnum#7=1\draw[thick] (x) --(y);\fi{}
    \ifnum#7=2\draw[thick] (z) --(y);\fi{}
    \ifnum#7=3\draw[thick] (x) -- (z);\fi{}
    \ifnum#7=4\draw[thick] (z) --(y) -- (x) -- (z);\fi{}
  }
  
  \def\radius{2}
  \cake{1,2,...,6}{0}{0}{1}{3}{5}{4}
   \node at (\xsh,\radius+0.5){$G$};
  \foreach \x in {0,1,...,5}{
    \draw[scale=1.5,->,>=stealth] (1.5,0) to (2.25,0.5-0.2*\x);
  }
  \def\radius{1.5}
  \cake{1,2,3}{5}{1.75}{1}{3}{5}{1}
  \cake{1,2,4}{5}{-1.75}{1}{3}{5}{0}
  
  \node at (5+2.5,0)[scale=1.5]{$\dots$};
  \node at (5+5,1.75+\radius+0.35){$G[V_1\cup V_3\cup V_{5}]$};
  \cake{1,3,5}{5+1*5}{1.75}{1}{3}{5}{4}
  \cake{1,4,5}{5+1*5}{-1.75}{1}{3}{5}{3}
  
  \node at (5+7.5,0)[scale=1.5]{$\dots$};
  \cake{3,5,6}{5+2*5}{1.75}{1}{3}{5}{2}
  \cake{4,5,6}{5+2*5}{-1.75}{1}{3}{5}{0}

  \end{tikzpicture}
  \caption{Schematic illustration of the idea behind our diminisher for the parameter order of the largest connected component.} %
  \label{fig:NWT-diminisher}
\end{figure}

Employing~\cref{constr:constsizedim}, we obtain the following.

\begin{proposition}
  \label{prop:NWToderstrictker}
 NWT and~TC parameterized by the order~$\ell$ of the largest connected component admit a strong ($n+m$)-diminisher.
\end{proposition}
For the following two lemmas, let~$H$ be an arbitrary but fixed connected constant-size graph with~$c>1$ vertices and let~$G=(V,E)$ be a graph with order~$\ell$ of the largest connected component.

\begin{lemma}
  \label{lem:atmostsize}
  \cref{constr:constsizedim} outputs in $O(n+m)$ time a graph~$G'$ with connected components of order at most~$\max\{\ell/2,4c\}$ . 
\end{lemma}

\begin{proof}
  In the case of~$\ell>4c$, note that $\lfloor n_i/(4c)\rfloor\leq |V_i^p|\leq \lceil n_i/(4c)\rceil$ for all $p\in\{1,\ldots,4c\}$.
  Moreover, $|V_i^{a_1} \cup \ldots \cup V_i^{a_c}|\leq c\cdot \lceil n_i/(4c)\rceil \leq \ell/4+c< \ell/2$.
  The size of~$G'$ is~$O(\binom{4c}{c}(n+m)) = O(n+m)$ as~$c$ is constant.
  It is not difficult to see that~$G'$ can be constructed in~$O(n+m)$ time.
\end{proof}

\begin{lemma}
	\label{lem:isomorph}
	Graph~$G$ contains a subgraph~$F=(V_F,E_F)$ isomorphic to~$H$ if and only if~$G'$, returned by~\cref{constr:constsizedim}, contains a subgraph~$F'=(V_F',E_F')$ isomorphic to~$H$, where $V_F'$  and $E_F'$ are copies of $V_F$ and $E_F$ in~$G'$, respectively.
\end{lemma}

\begin{proof}
  Clearly, as $G'$ is a disjoint collection of induced subgraphs and~$H$ is connected, if~$G'$ contains a subgraph isomorphic to~$H$, then also~$G$ does.
  
  Let~$G$ contain a subgraph~$F$ isomorphic to~$H$.
  If $\ell\leq 4c$, then~$G'=G$ contains~$F$.
  Otherwise, if~$\ell>4c$, then consider the following two cases.
  If~$F$ is contained in a connected component in~$G$ of size at most~$\ell/2$, then~$F$ is also contained in~$G'$.
  Otherwise,~$F$ is contained in a connected component~$G_i$ of size larger than~$\ell/2$.
  Let~$V(F)\subseteq V_i^{a_1} \cup \ldots \cup V_i^{a_c}$ for some $\{a_1,\ldots,a_c\}\subseteq \binom{[4c]}{c}$ (recall that~$F$ contains~$c$ vertices).
  Then~$F$ is a subgraph of $G[V_i^{a_1} \cup \ldots \cup V_i^{a_c}]\subseteq G'$.
\end{proof}

With~$H$ being a triangle~($c=3$) while asking for negative weight, due to~\cref{lem:atmostsize,lem:isomorph}, we get a strong~$(n+m)$-diminishers for NWT.
When asking for a specific vertex-coloring, this also yields a strong~$(n+m)$-diminisher for~TC.
  
\begin{proof}[Proof of~\cref{prop:NWToderstrictker}]
  Given an edge-weighted graph~$G=(V,E,w)$, we apply~\cref{constr:constsizedim} to~$G$ with $H$~being a triangle (note that $c=3$) to obtain~$G'$.
  We introduce the edge-weights~$w'$ to~$G'$ by assigning for each edge~$e\in E$ its weight to all of its copies~$e'\in E(G')$.
  By~\cref{lem:atmostsize}, $G'$ is constructed in linear time.
  By~\cref{lem:isomorph} and the definition of~$w'$, $G'$ contains a negative weight triangle if and only if~$G$ does.
  Hence, this procedure is a strong linear-time diminisher with respect to the order~$\ell$ of the largest connected component, as (by \cref{lem:atmostsize}) either $\ell'\leq \ell/2$, or~$\ell'\leq 4c$ (implying~$G'=G$), where in the latter case our strong diminisher decides whether~$G'$ contains a triangle of negative weight in~$O(n)$~time.
  
  For TC, the proof works analogously except that for each vertex~$v\in V$, we color its copies in~$G'$ with the color of~$v$.
\end{proof}

There is a straight-forward~$O(k^2\cdot n)$-time algorithm for NWT and TC: Check for each vertex all pairs of other vertices in the same connected component. 
However, under the APSP-conjecture (and SETH for TC) there are no~$O(n^{3-\eps})$-time algorithms for any~$\eps > 0$~\cite{AWY15,WW10}. 
Combining this with our diminisher in \cref{prop:NWToderstrictker} we can exclude certain strict kernels as shown below.

\begin{proof}[Proof of~\cref{thm:mainklb}(i)]
  By~\cref{prop:NWToderstrictker}, we know that NWT admits a strong $(n+m)$-diminisher.
  Suppose that NWT admits a strict $(n^\alpha, k^\beta)$-kernel for~$\alpha \ge 1, \beta \ge 1$ with~$\alpha \cdot \beta =3-\eps_0$, $\eps_0>0$.
  It follows by~\cref{thm:kernel+dimin=alg} that NWT is solvable in~$t(n,k)\in O(k^{\beta\cdot\alpha}\log(k) +n^\alpha)$~time. 
  Observe that~$\log(k)\in O(k^{\eps_1})$ for~$0<\eps_1<\eps_0$.
  Together with~$k\leq n$ and $\alpha \cdot \beta=3-\eps_0$
  we get~$t(n,k)\in O(n^{3-\eps})$ with $\eps=\eps_0-\eps_1>0$.
  Hence, the APSP-conjecture breaks~\cite{WW10}.
  The proof for~TC works analogously.
\end{proof}

\subsection{Parameter Maximum Degree}

The diminisher described in~\cref{constr:constsizedim} does not necessarily decrease the maximum degree of the graph.
We thus adapt the diminisher to partition the edges of the given graph (using an (improper) edge-coloring) instead of its vertices.
Furthermore, if~$H$ is of order~$c$, then $H$ can have up to~$\binom{c}{2} \leq c^2$ edges.
Thus, our diminisher considers all possibilities to choose~$c^2$ (instead of~$c$) parts of the partition.
For the partitioning step, we need the following.

\begin{lemma}
  \label{lem:coloring}
	Let~$G = (V,E)$ be a graph with maximum degree~$\Delta$ and let~$b \in \N$.
	One can compute in~$O(b(n+m))$~time an (improper) edge-coloring~$\col\colon E \to \N$ with less than~$2b$ colors such that each vertex is incident to at most~$\lceil \Delta / b \rceil$ edges of the same color.
\end{lemma}

\begin{proof}
    The edge-coloring can be computed in~$O(b(n+m))$~time with a simple generalization of a folklore greedy algorithm to compute a proper edge-coloring ($b = \Delta$):
    Consider the edges one by one and assign each edge the first available color.
    Observe that at any considered edge each of the two endpoints can have at most~$b-1$ unavailable colors, that is, colors that are used on~$\lceil \Delta / b \rceil$ other edges incident to the respective vertex.
    Hence, the greedy algorithm uses at most~$2b-1$ colors.
    The algorithm stores at every vertex an array of length~$b-1$ to keep track of the number of edges with the respective colors.
    Thus, the algorithm can for each edge simply try all colors at each edge in $O(b)$~time.
    Altogether, this gives $O(b(n+m))$~time to compute the edge-coloring.
\end{proof}

\begin{constr}
 \label{constr:maxdegdeg}
  Let~$H$ be an arbitrary but fixed connected constant-size graph of order~$c>1$.
  Let~$G=(V,E)$ be a graph with maximum degree~$\Delta$.
  First, employ~\cref{lem:coloring} to compute an (improper) edge-coloring~$\col\colon E \to \N$ with~$4c^2\leq f < 8c^2$ many colors (without loss of generality we assume $\Im(\col)=\{1,\ldots,f\}$) such that each vertex is incident to at most~$\lceil \Delta / (4c^2) \rceil$ edges of the same color.
  
  Now, construct a graph~$G'$ as follows.
  Let~$G'$ be initially the empty graph.
  If~$\Delta\leq 4c^2$, then set~$G'=G$.
  Otherwise, if~$\Delta>4c^2$, then construct~$G'$ as follows.
  We first partition~$E$:
  Let $E^p$ be the edges of color~$p$ for every $p\in\{1,\ldots,f\}$.
  Clearly, $E=E^1\uplus\cdots\uplus E^{f}$.
  Then, for each~$\{a_1,\ldots,a_{c^2}\}\in\binom{[f]}{c^2}$, add the graph~$(V,E^{a_1} \cup \ldots \cup E^{a_{c^2}})$ to~$G'$.
  This completes the construction.
  \qedconstr
\end{constr}

\begin{proposition}
 \label{prop:dimwrtmaxdeg}
  NWT and~TC parameterized by maximum degree~$\Delta$ admit a strong ($n+m$)-diminisher.
\end{proposition}
For the following two lemmas, let~$H$ be an arbitrary but fixed connected constant-size graph with~$c>1$ vertices and let~$G=(V,E)$ be a graph with maximum degree~$\Delta$.

\begin{lemma}
  \label{lem:lintimemaxdeg}
  \cref{constr:maxdegdeg} outputs a graph~$G'$ in $O(n+m)$ time with maximum degree~$\Delta(G')\leq \max\{\Delta/2,4c^2\}$. 
\end{lemma}

\begin{proof}
      In the case of~$\Delta>4c^2$, each vertex is incident to at most~$\lceil \Delta/(4c^2)\rceil$ edges of~$E^p$ for all $p\in\{1,\ldots,f\}$.
      Thus, in~$(V,E^{a_1} \cup \ldots \cup E^{a_{c^2}})$ the maximum degree is at most $c^2 \cdot \lceil \Delta/(4c^2)\rceil \leq \Delta/4+c^2< \Delta/2$.
      Using \cref{lem:coloring} with~$b = 4c^2 \in O(1)$, it is not difficult to see that~$G'$ is constructed in~$O(n+m)$ time.
\end{proof}

\begin{lemma}
  \label{lem:equivmaxdeg}
  Graph~$G$ contains a subgraph~$F=(V_F,E_F)$ isomorphic to~$H$ if and only if~$G'$, returned by~\cref{constr:maxdegdeg}, contains a subgraph~$F'=(V_F',E_F')$ isomorphic to~$H$, where $V_F'$  and $E_F'$ are copies of $V_F$ and $E_F$ in~$G'$, respectively.
\end{lemma}

\begin{proof}
      Clearly, as $G'$ is a disjoint collection of subgraphs and~$H$ is connected, if~$G'$ contains a subgraph isomorphic to~$H$, then also~$G$ does. 
    Let~$G$ contain a subgraph~$F$ isomorphic to~$H$.
    If $\Delta\leq 4c^2$, then~$G'=G$ contains~$F$.
    Otherwise, if~$\Delta>4c^2$, then let~$E(F)\subseteq E^{a_1} \cup \ldots \cup E^{a_{c^2}}$ for some $\{a_1,\ldots,a_{c^2}\}\subseteq \binom{[f]}{c^2}$ (recall that~$F$ contains at most~$c^2$ edges).
    Then~$F$ is a subgraph of $(V,E^{a_1} \cup \ldots \cup E^{a_{c^2}})\subseteq G'$. 
\end{proof}

\begin{proof}[Proof of~\cref{prop:dimwrtmaxdeg}]
  Given an edge-weighted graph~$G=(V,E,w)$, we apply~\cref{constr:maxdegdeg} to~$G$ with $H$~being a triangle (note that $c=3$) to obtain~$G'$.
  We introduce the edge-weights~$w'$ to~$G'$ by assigning for each edge~$e\in E$ its weight to all of its copies~$e'\in E(G')$.
  By~\cref{lem:lintimemaxdeg}, $G'$ is constructed in linear time.
  By~\cref{lem:equivmaxdeg} and the definition of~$w'$, $G'$ contains a negative weight triangle if and only if~$G$ does.
  Hence, this procedure is a strong linear-time diminisher with respect to the maximum degree, as (by \cref{lem:lintimemaxdeg}) either $\Delta(G')\leq \Delta/2$, or~$\Delta(G')\leq 4c^2$, where in the latter case our strong diminisher decides whether~$G'$ contains a triangle of negative weight in~$O(n)$ time.
  
  For TC, the proof works analogously except that for each vertex~$v\in V$, we color its copies in~$G'$ with the color of~$v$.
\end{proof}

\subsection{Parameter Degeneracy}

The degeneracy of a graph is the smallest number~$d$ such that every induced subgraph contains a vertex of degree at most~$d$.
For parameter degeneracy, the diminisher follows the same idea as the diminisher for the parameter maximum degree (see~\cref{constr:maxdegdeg}).
The only difference between the two diminishers is how the partition of edge set is obtained.

\begin{constr}
 \label{constr:degeneracy}
  Let~$H$ be an arbitrary but fixed constant-size graph of order~$c>1$.
  Let~$G=(V,E)$ be a graph with degeneracy~$d$.
  First, compute a degeneracy ordering\footnote{This is an ordering of the vertices such that each vertex~$v$ has at most~$d$ neighbors ordered after~$v$.}~$\sigma$ in~$O(n+m)$~time~\cite{MB83}. 
  Construct a graph~$G'$ as follows.
  
  Let~$G'$ be initially the empty graph.
  If~$d\leq 4c^2$, then set~$G'=G$.
  Otherwise, if~$d>4c^2$, then construct~$G'$ as follows.
  First, for each vertex~$v\in V$, we partition the edge set~$E_v:=\{\{v,w\}\in E\mid \sigma(v)<\sigma(w)\}$ going to the right of~$v$ with respect to~$\sigma$ into $4c^2$ parts.
  Let $E_v$ be enumerated as $\{e_1,\ldots,e_{|E_v|}\}$.
  For each~$v$, we define~$E_v^p:=\{e_i\in E_v\mid i\bmod 4c^2=p-1\}$ for every $p\in[4c^2]$.
  Clearly, $E_v=E_v^1\uplus\cdots\uplus E_v^{4c^2}$.
  Next, we define $E^p:=\bigcup_{v\in V} E_v^p$ for every $p\in[4c^2]$.
  Clearly, $E=\biguplus_{1\leq p\leq 4c^2} E^p = \biguplus_{1\leq p\leq 4c^2} \biguplus_{v\in V} E_v^p$.
  Then, for each~$\{a_1,\ldots,a_{c^2}\}\in\binom{[4c^2]}{c^2}$, add the graph~$(V,E^{a_1} \cup \ldots \cup E^{a_{c^2}})$ to~$G'$.
  This completes the construction.
  \qedconstr
\end{constr}

\begin{proposition}
 \label{prop:dimwrtdegen}
  NWT and~TC parameterized by degeneracy admit a strong ($n+m$)-diminisher.
\end{proposition}

For the following two lemmas, let~$H$ be an arbitrary but fixed connected constant size graph of order~$c>1$ and let~$G=(V,E)$ be a graph with degeneracy~$d$.
  
\begin{lemma}
    \label{lem:lintimedegen}
    \cref{constr:degeneracy} outputs a graph~$G'$ in $O(n+m)$ time with degeneracy at most~$\max\{d/2,4c^2\}$. 
\end{lemma}

\begin{proof}
    In the case of~$d>4c^2$, for each $p\in[4c^2]$, the degeneracy of~$F:=(V,E^p)$ is at least $\lfloor d/(4c^2)\rfloor$ and at most~$\lceil d/(4c^2)\rceil$.
    To see this, consider $F$ with ordering~$\sigma$ on its vertices~$V(F)$.
    Then, for each~$v\in V(F)$, exactly~$\lfloor |E_v|/(4c^2)\rfloor \leq |E_v^p|\leq \lceil |E_v|/(4c^2)\rceil$ vertices~$w$ with~$\sigma(w)>\sigma(v)$ are incident with~$v$ in~$F$.
    As~$|E_v|\leq d$, the claim follows.
    Moreover, the degeneracy of~$(V,E^{a_1} \cup \ldots \cup E^{a_{c^2}})$ is at most $c^2\cdot \lceil d/(4c^2)\rceil \leq d/4+c^2< d/2$.
    It is not difficult to see that~$G'$ is constructed in~$O(n+m)$ time.
\end{proof}

\begin{lemma}
    \label{lem:equivdegen}
    Graph~$G$ contains a subgraph~$F=(V_F,E_F)$ isomorphic to~$H$ if and only if~$G'$, returned by~\cref{constr:degeneracy}, contains a subgraph~$F'=(V_F',E_F')$ isomorphic to~$H$, where $V_F'$  and $E_F'$ are copies of $V_F$ and $E_F$ in~$G'$, respectively.
\end{lemma}

\begin{proof}
      Clearly, as $G'$ is a disjoint collection of subgraphs, if~$G'$ contains a subgraph isomorphic to~$H$, then also~$G$ does.
    Let~$G$ contain a subgraph~$F$ isomorphic to~$H$.
    If $d\leq 4c^2$, then~$G'=G$ contains~$F$.
    Otherwise, if~$d>4c^2$, then let~$E(F)\subseteq E^{a_1} \cup \ldots \cup E^{a_{c^2}}$ for some $\{a_1,\ldots,a_{c^2}\}\subseteq \binom{[4c^2]}{c^2}$ (recall that~$F$ contains at most~$c^2$ edges).
    Then~$F$ is a subgraph of $(V,E^{a_1} \cup \ldots \cup E^{a_{c^2}})\subseteq G'$. 
\end{proof}
  
\begin{proof}[Proof of~\cref{prop:dimwrtdegen}]
  Given an edge-weighted graph~$G=(V,E,w)$, we apply~\cref{constr:degeneracy} to~$G$ with $H$~being a triangle (note that $c=3$) to obtain~$G'$.
  We introduce the edge-weights~$w'$ to~$G'$ by assigning for each edge~$e\in E$ its weight to all of its copies~$e'\in E(G')$.
  By~\cref{lem:lintimedegen}, $G'$ is constructed in linear time.
  By~\cref{lem:equivdegen} and the definition of~$w'$, $G'$ contains a negative weight triangle if and only if~$G$ does.
  Hence, this procedure is a strong linear-time diminisher with respect to degeneracy, as (by \cref{lem:lintimedegen}) either $d'\leq d/2$, or~$d'\leq 4c^2$, where in the latter case our strong diminisher decides whether~$G'$ contains a triangle of negative weight in~$O(n)$ time.
  
  For TC, the proof works analogously except that for each vertex~$v\in V$, we color its copies in~$G'$ with the color of~$v$.
\end{proof}

\section{(Turing) Kernelization Upper Bounds}
We complement our results on kernelization lower bounds by showing straight-forward strict kernel results for \textsc{$H$-Subgraph Isomorphism} for connected constant-size~$H$ to show the limits of any approach showing kernel lower bounds.

\paragraph*{Strict Turing Kernelization.}
For the parameters order of the largest connected component and maximum degree, we present strict $(a,b)$-Turing kernels:
\begin{definition}
 \label{def:TK}
 A \emph{strict $(a,b)$-Turing kernelization} for a parameterized problem~$L$ is an algorithm that decides every input instance~$(x,k)$ in time $O(a(|x|))$ given access to an oracle that decides whether $(x',k')\in L$ for every instance $(x',k')$ with $|x'|+k'\leq b(k)$ in constant time.
\end{definition}
Note that the diminisher framework in its current form cannot be applied to exclude (strict) $(a,b)$-Turing~kernelizations.
In fact, it is easy to see that \textsc{$H$-Subgraph Isomorphism} for connected constant-size~$H$ parameterized by the order~$\ell$ of the largest connected component admits an~$(n+m,\ell^2)$-Turing~kernel, as each oracle call is on a connected component (which is of size at most~$O(\ell^2)$) of the input graph.
We present a strict Turing~kernelization for \textsc{$H$-SI} for connected constant-size~$H$ parameterized by maximum degree~$\Delta$.

\begin{proposition}
  \label{prop:turker}
	\textsc{$H$-Subgraph Isomorphism} for connected $H$ with~$c=|V(H)|$ parameterized by maximum degree~$\Delta$ admits a strict $(n\cdot\Delta\cdot(\Delta-1)^{\floor{c/2}},\Delta\cdot(\Delta-1)^{\floor{c/2}})$-Turing~kernel.
\end{proposition}

\begin{proof}
    Let~$(G=(V,E))$ be an input instance of~\textsc{$H$-Subgraph Isomorphism} and let~$\Delta$ denote the maximum degree in~$G$.
    For each vertex~$v\in V$, we create the subgraph~$G_v$ that is the subgraph induced by the closed $\floor{c/2}$-neighborhood~$N_G^{\floor{c/2}}[v]$ of~$v$ (we refer to these as subinstances).
    It is not difficult to see that in each subinstance the graph is of size at most~$2\Delta\cdot(\Delta-1)^{\floor{c/2}}$ and each subinstance can be constructed in time linear in its size.
    The algorithm outputs \yes{} if and only if there is at least one subinstances containing~$H$.
    This results in a total running time of~$O(n\cdot\Delta\cdot(\Delta-1)^{\floor{c/2}})$.
    
    In the remainder, we prove that~$G$ contains~$H$ if and only if there exists a~$v\in V$ such that~$G_v$ contains~$H$.
    
    \noindent(\emph{if}) This direction is clear as~$G_v$ is an induced subgraph of~$G$ for every~$v\in V$.
    
    \noindent(\emph{only if})
    Recall that~$H$ is connected and~$c=|V(H)|$.
    Hence, there is a vertex~$u\in V(H)$ such that $\dist_H(u,w)\leq \floor{c/2}$ for every~$w\in V(H)$.
    Let~$v$ be the vertex in~$G$ that corresponds to~$u$ in~$H$.
    Then it is not difficult to see that~$G_v$ contains~$H$ as~$G_v$ is induced on all vertices in~$G$ that are of distance at most~$\floor{c/2}$ from~$v$.
\end{proof}

\paragraph{Running-time Related Strict Kernelization.}

For NP-hard problems, it is well-known that a decidable problem is fixed-parameter tractable if and only if it admits a kernel~\cite{DF13}.
In the proof of the \emph{only if}-statement, one derives a kernel of size only depending on the running time of a fixed-parameter algorithm solving the problem in question.
We adapt this idea to derive a strict kernel where the running time and size admit such running time dependencies.

\begin{theorem}
  \label{thm:pl-fpt-in-p-implies-kernel}
	Let $L$ be a parameterized problem admitting an algorithm solving each instance~$(x,k)$ in $k^c\cdot |x|$ time for some constant~$c>0$.
	Then for every~$\eps>0$, each instance~$(x,k)$ admits a strict~$(|x|^{1+c/(1+\eps)},k^{1+\eps})$-kernel.
\end{theorem}

\begin{proof}
  Let $\eps>0$ arbitrary but fixed.
  If $k^{1+\eps}\geq |x|$, then the size of the instance is bounded by $k^{1+\eps}+k$.
  Otherwise, if $k^{1+\eps}< |x| \iff k<|x|^{1/(1+\eps)}$,  we can compute a constant-size kernel (trivial \yes-/\no-instance) in $k^c\cdot |x|<|x|^{c/(1+\eps)}\cdot |x|=|x|^{1+c/(1+\eps)}$ time.
\end{proof}
NWT and TC are both solvable in~$O(k^2 \cdot n)$ time ($k$ being the order~$\ell$ of the largest connected component, the maximum degree~$\Delta$, or the degeneracy~$d$~\cite{ChibaN85}).
Together with~\cref{thm:pl-fpt-in-p-implies-kernel} gives several kernelization results for~NWT and~TC, for instance, with~$\eps = 2$:

\begin{corollary}\label{cor:nwt-kernel}
	NWT admits a strict~$(n^{5/3},d^3)$-kernel when parameterized by the degeneracy~$d$ of the input graph.
\end{corollary}
Note that the presented kernel is a strict~$(n^\alpha,d^\beta)$-kernel with $\alpha=5/3$ and~$\beta=3$.
As~$\alpha\cdot \beta=5$ in this case, there is a gap between the above kernel and the lower bound of~$\alpha\cdot \beta\geq3$ in~\cref{thm:mainklb}(iii).
Future work could be to close this gap.

\section{Conclusion}
\label{sec:concl}

We provided the first conceptual analysis of strict kernelization lower bounds for problems solvable in polynomial time.
To this end, we used and (slightly) enhanced the parameter diminisher framework~\cite{CFM11,FFHKMN16}.
Our results for \textsc{Negative Weight Triangle} and \textsc{Triangle Collection} rely on the APSP-conjecture and SETH, but these assumptions can be replaced with any running-time lower bound known for the problem at hand.
Indeed the framework is not difficult to apply and we believe that developing special techniques to design diminishers is a fruitful line of further research.

We point out that the framework excludes certain \emph{trade-offs} between kernel size and running time: the smaller the running time of the diminisher, the larger the size of the strict kernel that can be excluded.
However, the framework in its current form cannot be used to exclude the existence of \emph{any} strict kernel of polynomial size in even linear time.

In this work, we only considered parameters that we call \emph{dispersed} parameters, defined as follows. 
Let $G$ be an instance of a graph problem $L$, and let~$G_1, G_2, \ldots, G_p$ be its connected components, 
where $p\geq 1$. 
A parameter~$k$ of~$G$ is \emph{dispersed} if $k(G)$ (i.e.~the value of the parameter $k$ in the graph $G$) 
is \emph{equal} to $k(G_i)$ for \emph{at least one} connected subgraph $G_i$ of $G$. 
Otherwise, if $k(G)$ is \emph{larger} than $k(G_i)$ for \emph{every} connected subgraph $G_i$ of $G$, 
then we call~$k$ an \emph{aggregated} parameter. 
In our opinion, it is of independent interest to apply the (strong) diminisher framework to graph problems with aggregated parameters.
Note that such a classification into dispersed and aggregated parameters has not been studied previously. 

We close with one concrete challenge:
		Is there a (strong) diminisher for NWT or TC with respect to the (aggregated) parameter feedback vertex number?
		Note that the disjoint union operation that we use in all our diminishers in \cref{sec:lower-bounds} can increase this parameter. 

\paragraph{Acknowledgement.}

We thank Holger~Dell (Saarland University) for fruitful discussion on~\cref{sec:frameworks} and Rolf~Niedermeier for discussions leading to this work.

\bibliographystyle{plainnat}
\bibliography{klbfptinp-arxiv-2}%

\end{document}